\documentclass[pra,twocolumn,showpacs,showkeys]{revtex4}

\usepackage{amsfonts,amsmath,amssymb,amsthm}
\usepackage{bbm,bm}
\usepackage{graphics,graphicx}
\usepackage{latexsym}
\usepackage{color}
\usepackage{times}

\newcommand{\e}{\mathrm{e}}

\newcommand{\rv}{{\bm r}}
\newcommand{\ra}{\rangle}
\newcommand{\la}{\langle}

\newcommand{\w}{\mathrm{W}}
\newcommand{\s}[1]{\sqrt{#1}}
\newcommand{\ket}[1]{\vert #1 \ra}

\newcommand{\bk}[2]{\la #1 \vert #2 \ra}

\newcommand{\ov}[2]{\left\la #1 | #2 \right\ra}
\renewcommand{\(}{\left(}
\renewcommand{\)}{\right)}

\renewcommand{\t}{\theta}

\renewcommand{\o}{\otimes}

\renewcommand{\>}{\rangle}

\newcommand{\R}{\mathbb{R}}
\newcommand{\x}{\mathbf{x}}
\theoremstyle{plain} 

\newtheorem{theorem}{Theorem}
\newcommand{\thmref}[1]{Theorem~\ref{#1}}

\newcommand{\be}{\begin{equation}}
\newcommand{\ee}{\end{equation}}

\begin{document}

\title{Duality and the geometric measure of entanglement of general multiqubit
W states}

\author{Sayatnova Tamaryan}
\affiliation{Theory Department, Yerevan Physics Institute, Yerevan, 375036,
Armenia}

\author{Anthony Sudbery}
\affiliation{Department of Mathematics, University of York, Heslington, York,
YO10 5DD, U.K.}

\author{Levon Tamaryan}
\affiliation{Physics Department, Yerevan State University, Yerevan, 375025,
Armenia}
\affiliation{Theory Department, Yerevan Physics Institute, Yerevan, 375036,
Armenia}

\begin{abstract} We find the nearest product states for arbitrary generalised W
states of $n$ qubits, and show that the nearest product state is essentially
unique if the W state is highly entangled. It is specified by a unit vector in Euclidean $n$-dimensional space. We use this duality between unit vectors and highly entangled W states to find the geometric measure
of entanglement of such states.
\end{abstract}



\maketitle

\paragraph{Introduction.}Quantifying entanglement of multipartite pure states
presents a real challenge to physicists. Intensive studies are under way and
different entanglement measures have been proposed over the
years~\cite{ben-schum,ben-rain,shim-95,vedr-rip,negat,robust}. However, it is
generally impossible to calculate their value  because the definition of any
multipartite entanglement measure usually includes a massive optimization over
certain quantum protocols or states~\cite{woot,iso,tri}.

Inextricable difficulties of the optimization are rooted in a tangle of
different obstacles. First, the number of entanglement parameters grows
exponentially with the number of particles involved~\cite{lindenprl}. Second,
in the multipartite setting several inequivalent classes of entanglement
exist~\cite{w,four}. Third, the geometry of entangled regions of robust states
is complicated~\cite{shared}. All of these make the usual optimization methods
ineffective~\cite{shared,sud-09,wei-guh}. Concise and elegant tools are
required to overcome this problem.

A widely used measure for multipartite systems is the geometric measure of
entanglement $E_g$~\cite{wei-03}, i.e.\ the distance from the nearest product
state. For an $n$-part pure state $\psi$  it is defined as  $E_g(\psi)=-2\ln
g(\psi)$, where the maximal product overlap $g(\psi)$ is given by
$$g(\psi)=\max_{u_1,u_2,...,u_k}|\ov{\psi}{u_1u_2...u_k}|,$$
and the maximization is performed over all product states. The maximal product
overlap has many remarkable applications. Among them are: it singles out the
multipartite states applicable for perfect quantum teleportation and superdense
coding~\cite{shared}, it can create a generalized Schmidt decomposition for
arbitrary $n$-part systems~\cite{hig}, it identifies irregularity in channel
capacity additivity~\cite{wern}, it quantifies the difficulty of distinguishing
multipartite quantum states by local means~\cite{local}, it is a good
entanglement estimator for quantum phase transitions in spin
models~\cite{phas-tr}, it detects a one-parameter family of maximally entangled
states~\cite{maximal}, and it can be easily estimated in
experiments~\cite{guh-06}.

In what follows states with $g^2>1/2$ are referred to as slightly entangled,
states with $g^2<1/2$ are referred to as highly entangled and states with
$g^2=1/2$ are referred to as shared quantum states. In this Letter we show how
to calculate the maximal product overlap of an arbitrary W state~\cite{w}. The
method is to establish a one-to-one correspondence between highly entangled W
states and their nearest product
states.

Consider first generalized Greenberger-Horne-Zeilinger states~\cite{ghz}, i.e.\
states that can be written $|\mathrm{GHZ}\> = a|0\ldots0\> + b|1\ldots1\>$ in
some product basis. Such states are fragile under local decoherence, i.e.\ they
become disentangled by the loss of any one party, and they are not highly
entangled in the sense defined above. The geometric measure of these states is
computed easily since the maximal overlap simply takes the value of the modulus
of the larger coefficient, $|a|$ or $|b|$ \cite{shim-04}. Accordingly, the
nearest separable state is the product state with the larger coefficient. Thus
many generalized GHZ states with different maximal overlaps can have the same
nearest product state.

Consider now generalized W-states~\cite{par-08}, which can be written
\begin{equation}\label{0.w}
\ket{\w_n}=c_1\ket{100...0} + c_2\ket{010...0} + \cdots + c_n\ket{00...01}.
\end{equation}
Without loss of generality we consider only the case of positive parameters
$c_k$ since the phases of the coefficients $c_k$ can be eliminated by
redefinitions of local states $\ket{1_k},\,k=1,2,...,n$.
The states \eqref{0.w} are robust against decoherence~\cite{raz-02}, i.e.\ loss
of any $n-2$ parties still leaves them in a bipartite entangled state.
Surprisingly, if the state is slightly entangled, then we have the same
situation as for generalized GHZ states: the maximal overlap is the largest
coefficient and, as before, many states can have the same nearest product
state~\cite{toward}. However, the situation is changed drastically when the
state is highly entangled. The calculation of the maximal overlap in this case
is a very difficult problem and the maximization has been performed only for
relatively simple
systems~\cite{wei-03,shim-04,tri,hay-sym,toward,sud-09,mgbbb,zch}.

On the other hand, different highly entangled W-states have different nearest
product states. This makes it possible to map the W-state to its nearest
product state and quickly obtain its geometric measure of entanglement. More
precisely, we construct two bijections. The first one creates a map between
highly entangled $n$-qubit W states and $n$-dimensional unit vectors $\x$. The
second one does the same between  $n$-dimensional unit vectors and $n$-part
product states. Thus we obtain a double map, or {\it duality}, as
follows
\begin{equation}\label{bijec}
 \ket{\w_n}\; \leftrightarrow\; \x\; \leftrightarrow\;
\ket{u_1}\o\ket{u_2}\o\cdots\o\ket{u_n}.
\end{equation}

The main advantage of the map is that if one knows any of the three vectors,
then one instantly finds the other two.

\smallskip

\paragraph{Classifying map.}Now we prove a theorem that provides a basis for
the map.

\begin{theorem}  \label{Wtou} Let $\ket{\w_n}$ be an arbitrary W state
\eqref{0.w} with non-negative real coefficients $c_i$, and let
$\ket{u_1}\o\ket{u_2}\o\cdots\o\ket{u_n}$ be its nearest product state. Then
the phase of $\ket{u_k}$ can be chosen so that
 $$\ket{u_k}=\sin\t_k\ket{0}+\cos\t_k\ket{1},\, 0\leq\t_k\leq\frac{\pi}{2}, \,
k=1,2,...,n.$$
where
\begin{equation}\label{0.dircos}
\cos^2\t_1+\cos^2\t_2+\cdots+\cos^2\t_n=1.
\end{equation}
\end{theorem}

\begin{proof} The nearest product state is a stationary point for the overlap
with $\ket{\w_n}$, so the states $\ket{u_k}$ satisfy the nonlinear eigenvalue
equations \cite{hig,wei-03,tri}
\begin{equation}\label{0.stat-eq}
\bk{u_1u_2\cdots\widehat{u_k}\cdots u_n}{\w_n}=g\ket{u_k};\;k=1,2,\cdots,n
\end{equation}
where the caret means exclusion. We can choose the phase of $\ket{u_k}$ so that
$
\ket{u_k} = \sin\t_k\ket{0} + \e^{i\phi_k}\cos\t_k\ket{1},
$
and then \eqref{0.stat-eq} gives the pair of equations
\begin{subequations}\label{0.eqn}
\begin{equation}\label{0.eqcos}
c_k\prod_{j\neq k}\sin\t_j = g\e^{i\phi_k}\cos\t_k,
\end{equation}
\begin{equation}\label{0.eqsin}
\sum_{l\neq k}\e^{-i\phi_l}c_l\cos\t_l\prod_{j\neq k,l}\sin\t_j = g\sin\t_k.
\end{equation}
\end{subequations}
Eq. \eqref{0.eqcos} shows that $g\e^{i\phi_k}$ is real, so $\phi_k = -\arg(g)$
is independent of $k$. Then the modulus of the overlap $|\<u_1\cdots
u_n|\w_n\>|$ is independent of $\phi$, so we can assume that $\phi = 0$. Now
multiplying eq.(\ref{0.eqsin}) by $\sin\t_k$ and using eq.(\ref{0.eqcos}) gives
Eq.(\ref{0.dircos}).
\end{proof}
Thus the angles $\cos\t_k$ define a unit $n$-dimensional Euclidean vector $\x$. We can also define a length $r$ as follows. From Eq.(\ref{0.eqcos}) it follows
that the ratio $\sin2\t_k/c_k$ does not depend on $k$. If this ratio is non-zero we can define
\begin{equation}\label{0.rmod}
\frac{1}{r} \equiv \frac{\sin2\t_1}{c_1} = \frac{\sin2\t_2}{c_2} = \cdots =
\frac{\sin2\t_n}{c_n}.
\end{equation}

\smallskip

\paragraph{Highly entangled W states.}Equations (5) admit a trivial
solution $\sin2\t_k=0,\,k=1,2,\cdots,n$ and a special solution with nonzero
values of all sines. The trivial solution gives the largest coefficient of
$|\w_n\>$ for the
maximal overlap and is valid for slightly entangled states. We consider them
later and now focus on the special solutions. From Eq.(\ref{0.rmod}) it follows
that
\begin{equation}\label{1.cos}
\cos^2\t_k=\frac{1}{2}\(1\pm\sqrt{1-\frac{c_k^2}{r^2}}\),\;k=1,2,\cdots,n.
\end{equation}
The plus sign means that $\cos2\t_k > 0$. Then from Eq.(\ref{0.dircos}) it
follows that this is possible for at most one angle; specifically, we prove
that if $\cos2\t_k>0$ for some $k$, then $c_k$ is the largest coefficient in
Eq.(\ref{0.w}). Suppose $\cos2\t_k>0$ but $c_k$ is not the largest coefficient
and there exists a greater coefficient, say $c_l$. Then from Eq.(\ref{0.rmod})
it follows that $\sin2\t_l>\sin2\t_k>0$ and consequently
$|\cos2\t_l|<|\cos2\t_k|$. Now we rewrite Eq.(\ref{0.dircos}) as follows:
\begin{equation}\label{1.2theta}
-\cos2\t_1-\cos2\t_2-\cdots-\cos2\t_n=n-2.
\end{equation}
From $|\cos2\t_l|<|\cos2\t_k|$ and $\cos2\t_k>0$ it follows that
$-\cos2\t_k-\cos2\t_l<0$  which is in contradiction with Eq.(\ref{1.2theta}).
Thus $c_k$ must be the largest coefficient.

Without loss of generality we assume that $0 \le c_1 \le \cdots \le c_n$. Then
in \eqref{1.cos} we must take the $-$ sign for $k = 1,\ldots,n-1$ and
\eqref{0.dircos} becomes
\begin{equation}\label{1.r}
\s{1 - \frac{c_1^2}{r^2}}  + \cdots + \s{1 - \frac{c_{n-1}^2}{r^2}} \pm \s{1 -
\frac{c_n^2}{r^2}} =
n-2
\end{equation}
We will denote the left-hand sides of these equations as $f_\pm(r)$. We also
use $f_0(r)$ to denote this expression without the last term. The function
$r(c_1,c_2,...,c_n)$ defined by $f_+(r)=n-2$ is a completely symmetric function
of the state parameters  $c_k$. In contrast, the function defined by
$f_-(r)=n-2$  is an asymmetric function since its dependence on the maximal
coefficient $c_n$ is different. Thus in equation \eqref{1.r} the upper and
lower signs describe symmetric and asymmetric entangled regions of highly
entangled states, respectively.

For highly entangled states, eqs. \eqref{1.r}$_\pm$ uniquely define $r$ as a
function of the state parameters $c_k$. More precisely,
\begin{theorem}\label{rsolutions}
There are two critical values $r_1$ and $r_2$ of the largest coefficient $c_n$,
i.e.\ functions of $c_1,\ldots,c_{n-1}$ such that
\begin{enumerate}
\item If $c_n \le r_1$, there is a unique solution of (\ref{1.r}$_+$) and no
solution of (\ref{1.r}$_-$);
\item If $c_n = r_1$, both (\ref{1.r}$_+$) and (\ref{1.r}$_-$) have a unique
solution, the same for both;
\item If $r_1 < c_n \le r_2$, there is no solution of (\ref{1.r})$_+$ and a
unique solution of (\ref{1.r}$_-$);
\item If $c_n > r_2$, neither (\ref{1.r}$_+$) nor (\ref{1.r}$_-$) has a
solution. In this case the state $|\w_n\>$ is slightly entangled.
\end{enumerate}
\end{theorem}

The value $r_1$ is the solution of $f_0(r_1) = n-2$, which exists and is unique
since $f_0(c_{n-1}) < n-2$ and $f_0(r) \rightarrow n - 1$ monotonically as $r
\rightarrow \infty$; and $r_2$ is defined by
\be
r_2^2 = c_1^2 + \cdots + c_{n-1}^2.
\ee
Then $r_2 \ge r_1$, for $f_0(r_2)\ge n - 2 = f_0(r_1)$ using $\sqrt{x} \ge x$
for $0 \le x \le 1$. Since $f_0$ is an increasing function of $r,$ it follows
that $r_2 \ge r_1$. Now the theorem follows from the following properties of
the functions $f_\pm(r)$($f'_-$ is the derivative of $f_-$):

\smallskip
1. $f_0$ and $f_+$ are monotonically increasing functions of $r$.

2. $f_+(r) \rightarrow n$ as $r \rightarrow \infty$.

3. If $c_n \le r_1$, $f_+(c_n) = f_0(c_n) \le f_0(r_1) = n-2$.

4. If $c_n \ge r_1$, then $f_+(r) \ge n-2$ for all $r > r_1$.

5. If $c_n < r_1$, then $f_-(c_n) < n - 2$.

6. If $c_n > r_1$, then $f_-(c_n) > n - 2$.

7. If $c_n < r_2$, then $f_-(r) < n - 2$ for large $r$.

8. If $c_n > r_2$ then $f_-(r) > n - 2$ for large $r$.

9. $f'_-(c_n + \epsilon) < 0$ for small $\epsilon$.

${}$\!\!\!10. If $c_n > r_2$, then $f'_-(r) < 0$ for all $r \ge c_n$.

\smallskip

These properties are illustrated in Figure \ref{curves}.

\begin{figure}[ht!]
\begin{center}
\includegraphics[width=7cm]{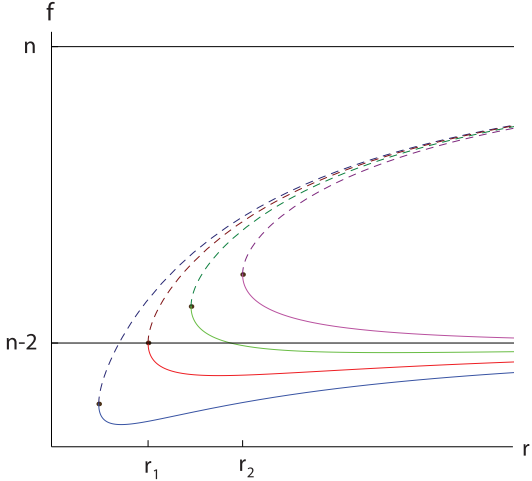}
\caption[fig1]{\label{curves}(Color online) The behaviour of the functions
$f_\pm$ for five-qubit W states. The function $f_+(r)$ (dotted line) and
$f_-(r)$ (solid line) are  plotted against $r$ in the four cases $c_n < r_1$,
$c_n = r_1$, $r_1 < c_n < r_2$ and $c_n = r_2$.}
\end{center}
\end{figure}

\paragraph{Geometric measure.} We can now identify the nearest product state,
and the largest product state overlap $g(|\w_n\>)$, for any W-state $|\w_n\>$,
as follows.

\begin{theorem} \label{nearestproduct} If $c_n \ge 1/2$, the state $|\w_n\>$
defined by \eqref{0.w} is slightly entangled. Its nearest product state is
$|0\ldots 01\>$, with overlap $g(|\w_n\>) = c_n$.

If $c_n \le 1/2$, the state $|\w_n\>$ is highly entangled and has nearest
product state
\be\label{product}
| u_1\>\ldots|u_n\> \;\;\text{ where }\;\; |u_k\> = \sin\t_k|0\> +
| \e^{i\phi}\cos\t_k|1\>,
\ee
with which its overlap is
\be \label{gdef}
g = 2r\sin\t_1\sin\t_2\ldots\sin\t_n.
\ee
Here $r$ is the solution of (\ref{1.r})$_\pm$, whose existence and uniqueness
are guaranteed by Theorem 2; the phase $\phi$ is arbitrary; and $\t_k$ is given
by \eqref{1.cos} with the $-$ sign for $k = 1,\ldots,n-1$, the $-$ sign for $k
= n$ if $r$ satisfies (\ref{1.r}$_+$), the $+$ sign if $r$ satisfies
(\ref{1.r}$_-$).
\end{theorem}
\begin{proof} The nonlinear eigenvalue equations \eqref{0.stat-eq} always have
$n$ solutions
\[
g = c_k,\qquad |u_i\> = \begin{cases} |0\> \text{ if } i\neq k,\\|1\> \text{ if
} i = k\end{cases}, \quad k = 1\ldots n
\]
If $c_n \ge \/2$, i.e.\ in case (4) of Theorem 2, there are no other stationary
values, so the largest overlap $g(|\w_n\>)$ equals the largest coefficient
$c_n$, the corresponding product state being $|0\ldots 01\>$.

If $c_n < 1/2$ there is another stationary value given by \eqref{gdef}. We will
now show that this is larger than any of the trivial stationary values $c_k$.
We use the following inequality: If $y_1,\ldots , y_n$ are real numbers lying
between 0 and 1, and satisfying $y_1 + \cdots + y_n \le 1$, then
\be \label{ineq1}
(1-y_1)(1-y_2)\cdots(1-y_n) \ge 1 - y_1 -y_2 - \ldots -y_n.
\ee
This is readily proved by induction.  We can apply \eqref{ineq1} to $n-1$ terms
of Eq.(\ref{0.dircos}) to get
$$(1-\cos^2\t_1)\cdots(1-\cos^2\t_{n-1})\ge 1-\cos^2\t_1 - \cdots -
\cos^2\t_{n-1}$$
or
\be
\sin^2\t_1\sin^2\t_2\sin^2\t_{n-1} \ge \cos^2\t_n.
\ee
Now from Eq.(\ref{0.eqcos}) it follows that $g^2 \ge c_n^2$. Thus $g$ is the
maximal product overlap, and the nearest product state is $|u_1\>\ldots u_n\>$.

Next we prove that if $|\w_n\>$ is normalised, then $g^2 < 1/2$. For this we
need another inequality: If $y_1,\ldots , y_n$ are real numbers lying between 0
and 1, and satisfying $y_1 + \cdots + y_n = n-1$, then
\be \label{ineq2}
y_1 + \cdots + y_n \ge y_1^2 + \cdots + y_n^2 + 2y_1y_2\ldots y_n.
\ee
This can also be proved by induction.

From \eqref{0.rmod}, and using $c_1^2 + \cdots + c_n^2 = 1$, we find
\be\label{rfromtheta}
r^2 = \frac{1}{\sin^2 2\t_1 + \cdots + \sin^2 2\t_n}.
\ee
Hence \eqref{gdef} gives
\be\label{gfromtheta}
g^2 = \frac{y_1y_2\ldots y_n}{y_1(1 - y_1) + \cdots y_n(1 - y_n)}
\ee
where $y_k = \sin^2\t_k$. But $y_1 + \cdots + y_n = n -1$, so the inequality
\eqref{ineq2} applies, and gives $g^2 \le 1/2$.
\end{proof}

Finally, we summarise the correspondence between highly entangled W-states,
their nearest product states, and unit vectors in $\R^n$.
\begin{theorem}\label{onetoone}
There is a 1:1 correspondence between highly entangled states $|\w_n\>$ defined
by \eqref{0.w}, their nearest product states with real non-negative
coefficients, and unit vectors $\x\in \R^n$ with $0<x_k<1/\sqrt{2}$ ($k =
1,\ldots, n-1$), $0 < x_n < 1$.
\end{theorem}
\begin{proof} By Theorem 3, $|\w_n\>$ is highly entangled if and only if $c_n <
1/2$. If this is the case, \thmref{Wtou} and \eqref{1.cos} show that its
nearest product state is of the form \eqref{product} where $\x = (\cos\t_1,
\ldots, \cos\t_n)$ is a unit vector in $\R^n$ in the region stated. The angles
$\t_k$ are given in terms of the coefficients $c_k$ by \eqref{0.rmod}, in which
$r$ is a function of the coefficients which, by \thmref{rsolutions}, is
uniquely defined. The nearest product states $|u_1\>|u_2\>\ldots|u_n\>$ are
determined by these angles, up to a phase $\phi$, by $|u_k\> = \sin\t_k|0\> +
\e^{i\phi}\cos\t_k|1\>$, so there is only one nearest product state with real
non-negative coefficients, and only one unit vector $\x$, for each highly
entangled state $|\w_n\>$. Conversely, given a unit vector $\x =
(\cos\t_1,\ldots,\cos\t_n)$, the quantity $r$ is determined by
\eqref{rfromtheta}, and then the coefficients $c_1,\ldots,c_n$ are determined
by \eqref{0.rmod}. Thus the correspondences \eqref{bijec} are bijections.
\end{proof}

The equations (\ref{1.r}$_\pm$) cannot always be explicitly solved to give
analytic expressions for $r$ in terms of the coefficients $c_k$. However, in
some cases, including all states for $n=3$, explicit solutions can be obtained.
Then the angles $\t_k$ can be calculated from \eqref{0.rmod} and
eq.\eqref{gdef} gives a formula for the maximal product overlap $g(|\w_n\>)$.
This formula is valid unless any of the angles $\t_k$ vanishes, and restores
all known results for the maximal overlap of highly entangled W states. When
$n=3$ it coincides with the formula (31) in Ref.\cite{tri}. When
$c_1=c_2=\cdots=c_n$  it coincides  with the formula (52) in
Ref.\cite{shim-04}. And when $n=4$ and $c_3=c_4$ it coincides with the formula
(37) derived in Ref.\cite{toward}.

When $ \max(c_1^2,c_2^2,\cdots c_n^2)=r_2^2=1/2$ the two expressions for
$g(|\w_n\>)$ given in \thmref{nearestproduct} coincide; these states are shared
quantum states. The nearest product states and maximal overlaps of shared
states are given by the first case of \thmref{nearestproduct}, but also they
appear as asymptotic limits of the second case. Indeed, at the limit $\t_n\to0$
we have
\begin{equation}\label{2.app}
\lim_{\t_n\to0}2r\sin\t_n\to c_n,\; \lim_{\t_n\to0}2r\cos\t_k\to c_k,\,k\neq n.
\end{equation}
Thus the angle $\t_n$ vanishes and the length of the vector $\rv$ goes to
infinity, but their product has a finite limit. Substituting these limits into
Eq.(\ref{0.dircos}) one obtains $c_n^2\to r_2^2$. Therefore entangled regions
of highly and slightly entangled states are separated by the surface
$c_n^2=1/2$; for states on the surface, $r\to\infty$. All of these states can
be used as a quantum channel for the perfect teleportation and superdense
coding~\cite{shared}.

\smallskip

\paragraph{Summary.}We have constructed correspondences between W states,
$n$-dimensional unit vectors and separable pure states. The map reveals two
critical values for quantum state parameters. The first critical value
separates symmetric and asymmetric entangled regions of highly entangled
states, whiles the second one separates highly and slightly entangled states.
The method gives an explicit expressions for the geometric measure when the
state allows analytical solutions, otherwise it expresses the entanglement as
an implicit function of state parameters.

It should be noted that the bijection between W states and $n$-dimensional unit
vectors is not related directly to the geometric measure of entanglement.
Therefore it is possible to extend the method to other entanglement measures.
To this end one creates an appropriate bijection between unit vectors and
optimization points of an entanglement measure one wants to compute. This work is in progress.
\begin{acknowledgments}
This work was supported by ANSEF grant PS-1852.
\end{acknowledgments}

\end{document}